\documentclass[10pt,conference]{IEEEtran}

\usepackage{times,amsmath,amssymb}
\usepackage{graphicx}
\usepackage{balance}  
\usepackage{type1cm}     
\usepackage[vlined,linesnumbered,ruled,noend]{algorithm2e}
\usepackage{xspace}     
\usepackage{placeins}
\usepackage{booktabs}     
\usepackage[bf,tableposition=top]{caption}     
\usepackage{siunitx}          
\usepackage[hyphens]{url}     
\usepackage{bold-extra}     

\usepackage{soul}
\usepackage{units}     
\usepackage{mathtools}     
\usepackage{multirow}
\DeclareGraphicsExtensions{.pdf}
\usepackage[bookmarks=false]{hyperref}
    
\usepackage[square,numbers]{natbib}     
\newcommand{\spara}[1]{\smallskip\noindent\textbf{#1}}

\newcommand{\tuple}[1]{\ensuremath{\langle #1 \rangle}\xspace}
\newcommand{\hash}[1]{\ensuremath{\mathcal{F}_{#1}}\xspace}

\newcommand{\dspe}{\textsc{dspe}\xspace}
\newcommand{\dspes}{{\dspe}s\xspace}
\newcommand{\dagr}{\textsc{dag}\xspace}

\newcommand{\pkg}{\textsc{Partial Key Grouping}\xspace}
\newcommand{\pkgs}{\textsc{pkg}\xspace}
\newcommand{\kg}{\textsc{kg}\xspace}
\newcommand{\sg}{\textsc{sg}\xspace}
\newcommand{\dc}{\textsc{d-c}\xspace}
\newcommand{\rr}{\textsc{rr}\xspace}
\newcommand{\wc}{\textsc{w-c}\xspace}
\newcommand{\sources}{\ensuremath{\mathcal{S}}\xspace}
\newcommand{\numsources}{\ensuremath{s}\xspace}
\newcommand{\workers}{\ensuremath{\mathcal{W}}\xspace}
\newcommand{\numworkers}{\ensuremath{n}\xspace}
\newcommand{\keyspace}{\ensuremath{\mathcal{K}}\xspace}
\newcommand{\keydistr}{\ensuremath{\mathcal{D}}\xspace}
\newcommand{\keysize}{\ensuremath{\lvert\keyspace\rvert}\xspace}
\newcommand{\nummsgs}{\ensuremath{m}\xspace}
\newcommand{\rank}[1]{\ensuremath{r(#1)}\xspace}
\newcommand{\hh}{\ensuremath{\mathcal{H}}\xspace}
\newcommand{\tail}{\ensuremath{\keyspace \setminus \hh}\xspace}
\newcommand{\naturals}{\mathbb{N}}
\newcommand{\mycomment}[1]{}
\newcommand{\code}[1]{{\textsc #1}}

\DeclareMathOperator*{\expect}{\mathbb{E}}
\DeclareMathOperator*{\avg}{avg}
\let\max=\undefined
\DeclareMathOperator*{\max}{max\vphantom{g}} 
\DeclareMathOperator*{\argmin}{argmin}

\newenvironment{squishlist}
{\begin{list}{$\bullet$}
  { \setlength{\itemsep}{0pt}
     \setlength{\parsep}{3pt}
     \setlength{\topsep}{3pt}
     \setlength{\partopsep}{0pt}
     \setlength{\leftmargin}{1.5em}
     \setlength{\labelwidth}{1em}
     \setlength{\labelsep}{0.5em} } }
{\end{list}}

\newtheorem{theorem}{Theorem}[section]
\newtheorem{proposition}[theorem]{Proposition}

\title{When Two Choices Are not Enough:\\Balancing at Scale in Distributed Stream Processing}

\author{%
{Muhammad Anis Uddin Nasir{\small $^{\sharp}$},
Gianmarco De~Francisci~Morales{\small $^{\diamond}$},
Nicolas Kourtellis{\small $^{\ddagger}$},
Marco Serafini{\small $^{\diamond}$} }
\vspace{1.6mm}\\
\fontsize{10}{10}\selectfont\itshape
$^{\sharp}$KTH Royal Institute of Technology, Stockholm, Sweden\\
$^{\ddagger}$Telefonica Research, Barcelona, Spain\\
$^{\diamond}$Qatar Computing Research Institute, Doha, Qatar\\
\fontsize{9}{9}\selectfont\ttfamily\upshape
anisu@kth.se,
gdfm@acm.org,
nicolas.kourtellis@telefonica.com,
mserafini@qf.org.qa
}
\begin{document}
\maketitle

\begin{abstract}
Carefully balancing load in distributed stream processing systems has a fundamental impact on execution latency and throughput.
Load balancing is challenging because real-world workloads are skewed: some tuples in the stream are associated to keys which are significantly more frequent than others.
Skew is remarkably more problematic in large deployments: having more workers implies fewer keys per worker, so it becomes harder to ``average out" the cost of hot keys with cold keys.

We propose a novel load balancing technique that uses a heavy hitter algorithm to efficiently identify the hottest keys in the stream.
These hot keys are assigned to $d \geq 2$ choices to ensure a balanced load, where $d$ is tuned automatically to minimize the memory and computation cost of operator replication.
The technique works online and does not require the use of routing tables.
Our extensive evaluation shows that our technique can balance real-world workloads on large deployments, and improve throughput and latency by $\mathbf{150\%}$ and $\mathbf{60\%}$ respectively over the previous state-of-the-art when deployed on Apache Storm.
\end{abstract}

\section{Introduction}
\label{sec:intro}

\enlargethispage{\baselineskip}

Stream processing is currently undergoing a revival,\footnote{\url{http://radar.oreilly.com/2015/04/a-real-time-processing-revival.html}}
mostly as a result of recent advances in Distributed Stream Processing Engines (\dspes) such as Storm, S4, and Samza.
The capabilities of modern engines even allow to take a unified approach to streaming and batch computations, such as in Flink~\cite{alexandrov2015stratosphere} and Google Dataflow~\cite{akidau2015dataflow}.
These technologies promise high efficiency combined with high throughput at low latency. 

The versatility achieved by \dspes has enabled developers to build a plethora of streaming applications.
Examples range from continuous query processing to online machine learning, from ETL pipelines to recommender systems~\citep{deFrancisciMorales2013samoa, deFrancisciMorales2015samoa}.
While the data streams that flow through these systems originate from disparate application domains, 
they have one common denominator: many real-world data streams are highly skewed.
Skew is one of the main obstacles to scaling distributed computation, mainly because it creates load imbalance, and the primary focus of this work.

To better appreciate the problem and its context, we first need to understand how \dspe applications are structured.
Streaming applications are commonly represented in the form of directed graphs that represent the \emph{data flow} of the application.
The vertices of the graph are data transformations (operators), and its edges are channels that route data between operators.
The data flowing along these edges is a stream, represented as a sequence of \emph{tuples}, each associated with a \emph{key}.
To achieve high throughput, modern distributed engines leverage data parallelism by creating several instances of an operator (workers) that process independent sub-streams.
These sub-streams are created by partitioning the incoming stream across instances of the downstream operator via \emph{grouping schemes}.

Among the existing grouping schemes, \emph{key grouping} is often used for stateful operators.
This scheme guarantees that all the tuples with the same key are processed by the same operator instance, akin to MapReduce.
Key grouping is commonly implemented via hashing, which is known to cause load imbalance across multiple workers, especially in presence of skew in the input key distribution~\cite{azar1999balanced-allocations, mitzenmacher2001power}.

\begin{figure}[t]
\begin{center}
\includegraphics[width=\columnwidth]{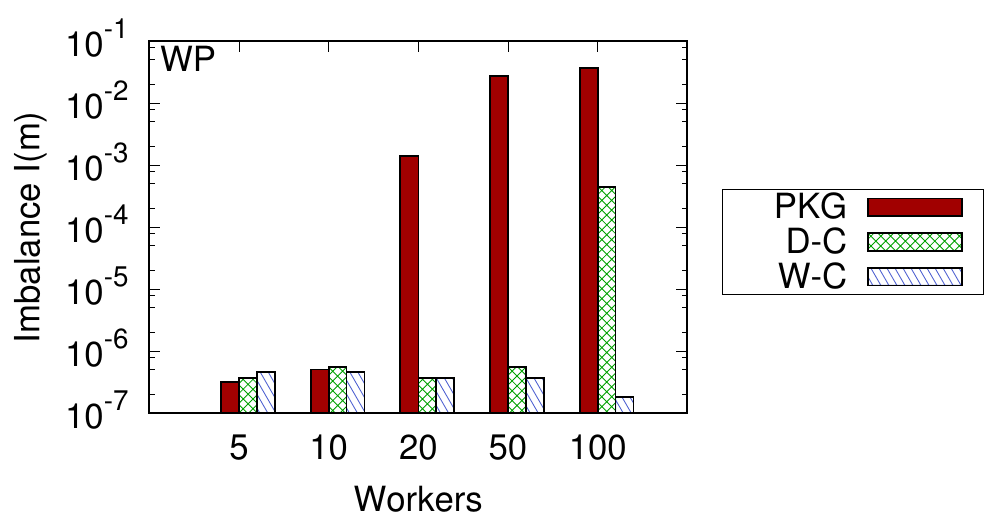}
\caption{Imbalance due to skew is more challenging to handle at large scale. On this dataset from Wikipedia, PKG is able to achieve low imbalance only at small scales, 
while the techniques proposed in this paper, D-Choices (D-C) and W-Choices (W-C), fare better at large scales.
}
\label{fig:pkg-fail}
\end{center}
\vspace{-2em}
\end{figure}
\enlargethispage{\baselineskip}

Some solutions to balance load in \dspes in presence of skew have recently been proposed~\cite{nasir2015power,gedik2014partitioning,rivetti2015efficient,nasir2015pkg}.
In particular, randomized techniques based on the ``power of two choices"~\cite{azar1999balanced-allocations}, such as \pkg~\cite{nasir2015power}, have found their way into mainstream \dspes such as Apache Storm, thanks to their combined simplicity and effectiveness.
Unlike key grouping, \pkgs associates each key to two possible operator instances, and selects the least loaded of the two whenever a tuple for a given key must be processed.
Techniques leveraging the power of two choices can adapt dynamically to workloads, they do not require the \dspe to support operator migration, and route keys to workers using hash functions rather than routing tables, which can become very large and must be replicated consistently across multiple upstream operators.

In this paper, we tackle a serious scalability problem of these techniques: using only two choices is not sufficient in very large deployments.
The source of the problem is that these techniques require the amount of skew in the distribution of keys to be inversely proportional to the size of the deployment.
For example, \pkgs requires the fraction of load generated by the most frequent key to be less than the combined ideal (i.e., average) load supported by two workers.
As the size of the deployment grows, the average fraction of load per worker decreases, so the fraction of load generated by the most frequent key must also decrease for the approach to be effective.
Intuitively, the imbalance exists because at larger scales it becomes increasingly difficult to average out the load among multiple keys.
Figure~\ref{fig:pkg-fail} shows this problem occurring while processing a real-world dataset from Wikipedia.
While at small scales (5 or 10) \pkgs manages to keep the system balanced, at larger scales (20, 50 or 100 workers) the imbalance grows significantly and gets close to 10\% (Section~\ref{sec:preliminaries} formally defines the notion of imbalance).

The scalability problem is exacerbated for workloads that are extremely skewed, as for example in graph streaming applications where some vertices are much more popular than others.\footnote{See \url{http://konect.uni-koblenz.de/statistics/power} for a repository of common graph datasets and their power law exponent.} 
Consider, for example, that under a Zipf distribution (which models the occurrence of words in natural language) with exponent $z = 2.0$, the most frequent key represents nearly $60\%$ of the occurrences, and thus \pkgs is unable to guarantee ideal load balance for any deployment larger than $3$ workers, as shown in the original analysis~\cite{nasir2015power}.


The straightforward alternative to using two choices is \emph{shuffle grouping}, a grouping scheme which distributes the messages in a round-robin fashion irrespective of the keys.
Shuffle grouping is the best choice for stateless operators.
However, it comes with a high cost when used for stateful operators: it is necessary to replicate the state associated with keys potentially on each worker, since each worker could have to process each key.
The memory overhead of shuffle grouping can thus become directly proportional to the number of workers used, therefore hindering scalability.

We propose two new streaming algorithms for load balancing, called \emph{D-Choices} and \emph{W-Choices}, that enable \dspes\ to scale to large deployments while supporting stateful operators.
As shown in Figure~\ref{fig:pkg-fail}, \emph{D-Choices} and \emph{W-Choices} achieve very low imbalance (smaller than $0.1\%$) even on very large clusters.
Our algorithms reduce the memory footprint by estimating the minimum number of choices per key for cases where two are not enough.
The intuition behind our techniques is that, for highly skewed inputs, most of the load on the system is generated by a handful of ``hot'' keys.
Thus, while the long tail of low-frequency keys can be easily managed with two choices, the few elements in the head needs additional choices.
Using more than two choices requires answering two main questions: ($i$) \emph{``which are the keys that need more than two choices?''} and ($ii$) \emph{``how many choices are needed for these keys?''}.

Our proposed method identifies and partitions the most frequent keys in the distribution, which are the main reason for load imbalance at scale.
It uses a streaming algorithm to identify ``heavy hitter'' keys that require a larger number of workers,
and allows them a number of choices $d \geq 2$ which is minimal yet sufficient for load balancing.
For the rest of the keys it uses the standard two choices.
To find the heavy hitters, we leverage the well-known SpaceSaving algorithm
~\cite{metwally2005spacesaving}, and its recent generalizations to the distributed setting~\cite{berinde2010space}.

The threshold used to find the heavy hitters is an important parameter in determining the memory cost of the approach, since these keys are mapped to a larger number of workers.
Our evaluation examines a range of potential thresholds and finds that a single threshold is sufficient for most settings.
This threshold requires only a very small number of keys to be assigned to more than two workers.
In fact, if the stream is very skewed the hot keys are only a handful, otherwise if the stream is not very skewed, two choices are sufficient for most keys, even for some of the most frequent ones.

Our two algorithms differ in the number of choices provided for the hot keys.
W-Choices allows as many choices as the number of workers.
Since the heavy hitter algorithm selects only few hot keys, this solution performs well if the state associated with each key is not extremely large.
D-Choices is a more sophisticated strategy that tries to estimate the minimum number of choices $d$ for hot keys that is sufficient to ensure a balanced load.
The estimation considers several factors, such as the relative frequency of the hot keys, and the overlaps among the sets of choices of different keys. 
In addition, since the mapping of keys to workers responsible for processing them is typically done using hash functions, we need to take into account their collision probability. 
Our analysis provides a lower bound on the expected value of $d$ which is very accurate.
We verify our bound via an extensive evaluation that compares the value of $d$ obtained from the analysis with the optimal value of $d$ found empirically in several simulated settings.

We evaluate our methods on real workloads and show that they allow scaling to large deployments, whereas existing approaches using only two choices do not scale to 50 or more workers.
This increased scalability translates into an improvement of up to $150\%$ in throughput, and up to $60\%$ in latency, over \pkgs when deployed on a real Apache Storm cluster.
The improvements are even larger over key grouping: $230\%$ in throughput and $75\%$ in latency.

In summary, this paper makes the following contributions:

\begin{squishlist}
\item We propose two new algorithms for load balancing of distributed stream processing systems; the algorithms are tailored for large-scale deployments and suited to handle the skew in real-world datasets;
\item We provide an accurate analysis of the parameters that determine the cost and effectiveness of the algorithms;
\item An extensive empirical evaluation\footnote{Code available at \url{https://github.com/anisnasir/SLBSimulator} and \url{https://github.com/anisnasir/SLBStorm}} shows that our proposal is able to balance the load at large scale and in the presence of extreme skew;
\item The reduced imbalance translates into significant improvements in throughput and latency over the state-of-the-art when deployed on Apache Storm.
\end{squishlist}



\section{Preliminaries}
\label{sec:preliminaries}

\subsection{The Dataflow Model}
Most modern \dspes run on clusters of machines that communicate by exchanging messages.
These \dspes run  streaming applications that are usually represented in the form of a directed acyclic graph (\dagr), which consists of a set of vertices and a set of edges.
The vertices in the \dagr (called \emph{operators}) are a set of data transformations that are applied on an incoming data stream, whereas the edges are data channels between vertices which connect the output of one operator to the input of the following one.
In order to achieve high performances, \dspes run multiple \emph{instances} of these operators and let each instance process a portion of the input stream, i.e., a \emph{sub-stream} (Figure~\ref{fig:skew-imbalance}).

For simplicity, we specify our notation by considering a single stream between a pair of operators, i.e., a single edge connecting two vertices.
Given a stream under consideration, let the set of upstream operator instances be called \emph{sources} and be denoted as \sources, and the set of downstream operator instances be called \emph{workers} and denoted as \workers, and their sizes be $\lvert \sources \rvert = \numsources$ and $\lvert \workers \rvert = \numworkers $.
Although the number of messages is unbounded, at any given time we denote the number of messages in the input stream as $m$.


The input stream is fed into a \dspe as a sequence of messages $\tuple{t,k,v}$, where $t$ is a timestamp, $k$ is a key, and $v$ is a value.
The keys in the messages are drawn from a (skewed) distribution \keydistr over a finite key space \keyspace, i.e., some of the keys appear more frequently than others.
Let $p_{k}$ be the probability of drawing key $k$ from \keydistr.
We define the \emph{rank} of a key \rank{k} as the number of keys that have a higher or equal probability of appearing according to \keydistr (with ties broken arbitrarily).
We identify the key by its rank in \keydistr, and therefore, by definition 
the probability of keys being drawn decreases as follows: $p_1 \geq p_2 \geq \ldots \geq p_{\keysize}$, with $\sum_{k \in \keyspace} p_k = 1$.

\subsection{Stream Partitioning}
Sub-streams are created by partitioning the input stream via a \emph{stream partitioning} function $P_{t}: \keyspace \to \naturals$, which maps each key in the key space to a natural number, at a given time $t$.
This number identifies the worker responsible for processing the message (we identify the set \workers of workers with a prefix of the naturals).
Each worker is associated to one or more keys, and keys associated to the same worker are said to be \emph{grouped} together by a \emph{grouping scheme}.
Henceforth, we consider three already existing grouping schemes:
\begin{squishlist}
\item \textbf{\textit{Key Grouping (\kg)}}:
This scheme ensures that messages with the same key are handled by the same downstream worker.
\kg is usually implemented via hashing.
\item \textbf{\textit{Shuffle Grouping (\sg)}}:
This scheme evenly distributes messages across the available downstream workers, thus ensuring ideal load balance.
\sg is usually implemented via round-robin selection of workers.
\item \textbf{\textit{Partial Key Grouping (\pkgs)}}~\cite{nasir2015power}:
The scheme ensures that messages with the same key are processed by at most two workers. 
It is implemented by using two separate hash functions that produce two candidate workers.
The message is then routed to the least loaded of the two workers.
\end{squishlist}

\spara{Load imbalance.}
We use a definition of \emph{load} similar to others in the literature (e.g., Flux~\citep{shah2003flux} and~\pkgs~\cite{nasir2015power}).
At time $t$, the load of a worker $w$ is the fraction of messages handled by the worker up to $t$:
\[
L_w(t) = \frac { \left\vert \left\{ \tuple{\tau,k,v} \mid P_{\tau}(k) = w \wedge \tau \leq t \right\} \right\vert }{ \nummsgs } .
\]

We define the \emph{imbalance} at time $t$ as the difference between the maximum and the average load of the workers:
\[
I(t) = \max_{w}{L_w(t)} - \avg_{w}{L_w(t)}, \quad w \in \workers .
\]
Our goal is to find a stream partitioning function $P_t$ that minimizes the maximum load $L(\nummsgs)$ of the workers, which is equivalent to minimizing the imbalance $I(\nummsgs)$ since the average load $\nicefrac{1}{\numworkers}$ is independent of $P_t$.

\begin{figure}[t]
\begin{center}
\includegraphics[width=\columnwidth]{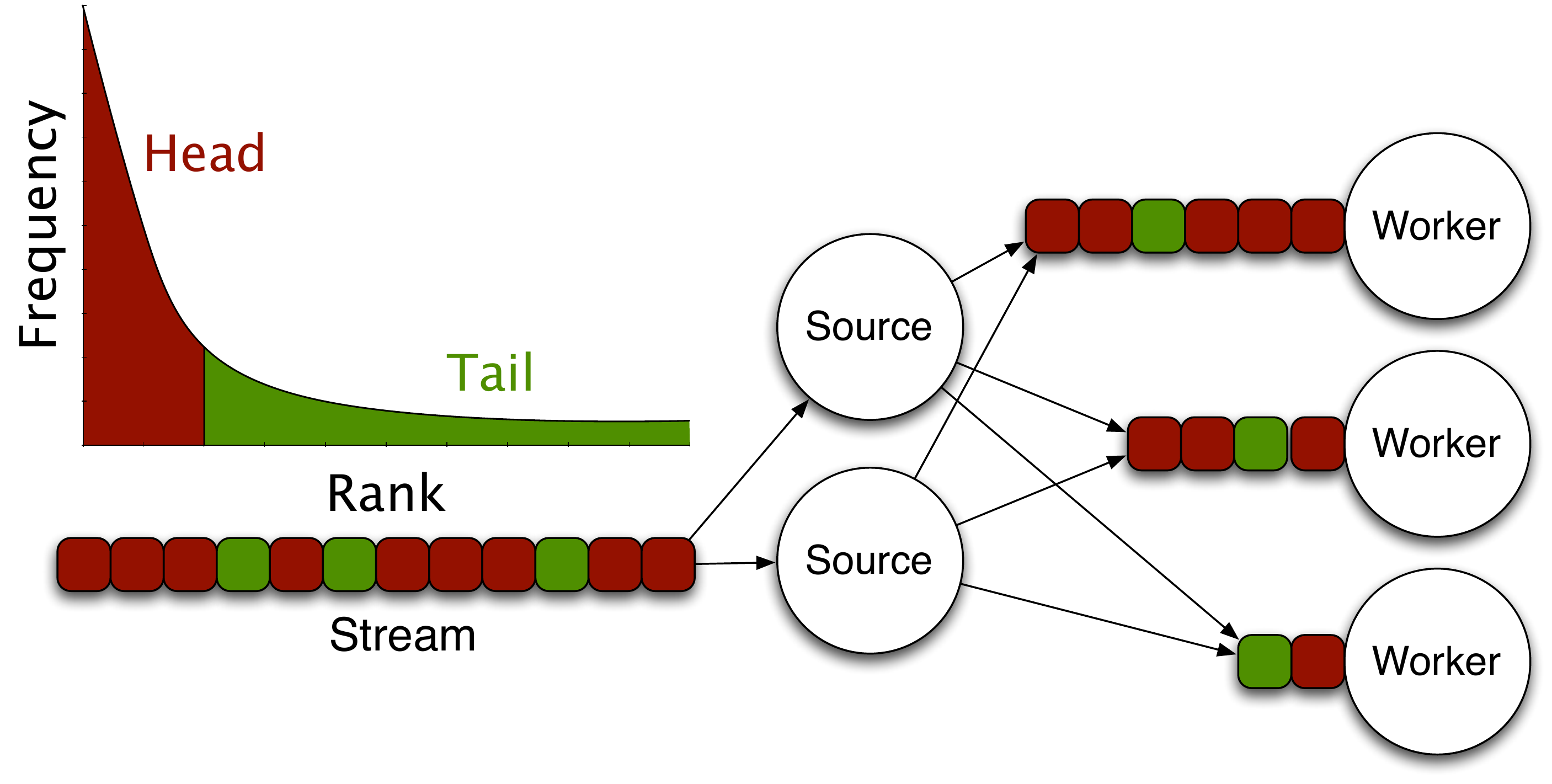}
\caption{Imbalance generated by high skew. Most of the load is due to the head of the distribution.}
\label{fig:skew-imbalance}
\end{center}
\vspace{-2em}
\end{figure}

\section{Solution Overview}
\label{sec:solution-overview}

The solution we explore in this paper revolves around the idea of detecting at runtime the hot keys or ``heavy hitters'', which comprise the \emph{head} of the key distribution, and treat them differently from the \emph{tail} when it comes to mapping them on workers for processing.
In particular, we track the head \hh of the key distribution in a distributed fashion across sources, and then split the keys belonging to the head across $d \geq 2$ workers.

To instantiate this solution, we need to define the size of the head (i.e., the cardinality of \hh) and the number of choices per key that  should be used (i.e., the value of $d$).
These two quantities jointly affect the load on the workers.
Given a skewed key distribution, there will be a combination of number of keys considered $|\hh|$ and number of choices $d$, that minimize the overall load imbalance across workers, while adding the smallest possible overhead.
However, such an ideal setup is difficult to reach at runtime, in a real system that analyzes streams of data with an unknown key distribution.

In order to tackle this multi-objective problem, we first define a threshold $\theta$ that separates the head from the tail of the distribution.
Then, we study the problem of finding $d$, the number of choices for the head that produces a low load imbalance.
Next, we discuss these two issues. 

\subsection{Finding the Head}
One important issue in our solution is how to define the threshold $\theta$ that splits the head of the distribution in the input stream from the tail.
%
%
%
%
%
%
Formally, we define the head \hh as the set of keys that are above a frequency threshold:
\[
\hh = \left\{ k \in \keyspace \mid p_{k} \geq \theta \right\} .
\]
As discussed in the introduction, algorithms based on the power of two choices make assumptions on the maximum skew in the key distribution. 
For instance, \pkgs assumes an upper bound on the frequency of the most frequent key.
This bound is inversely proportional to the total number of choices in the system, i.e., to the total number of workers in our setting.
We want to select the threshold $\theta$ such that the head includes all the keys violating the assumptions of \pkgs.
Therefore, we focus on the space where \pkgs falls short, according to its analysis.

The first bound provided by the analysis states that if $p_1 > \nicefrac{2}{\numworkers}$, then the expected imbalance at time $m$ is lower-bounded by
$(\frac{p_1}{2} - \frac{1}{\numworkers}) m$, which grows linearly with $m$~\citep{nasir2015power}.
This imbalance exists because the load generated by the most frequent key exceeds the capacity of two workers.
Therefore, clearly, all the keys that exceed the capacity of two workers belong to the head, and thus, $\theta \leq \nicefrac{2}{\numworkers}$.
The second bound in the analysis states that if $p_1 \leq \nicefrac{1}{5\numworkers}$, then the imbalance generated by \pkgs is bounded, with probability at least $1-\nicefrac{1}{\numworkers}$.
Therefore, we do not need to track keys whose frequency is smaller than this bound, and so $\theta \geq \nicefrac{1}{5\numworkers}$.

Given these theoretical bounds, we evaluate thresholds in the range $ \nicefrac{1}{5 \numworkers} \leq \theta \leq \nicefrac{2}{\numworkers}$.
As shown in Section~\ref{sec:evaluation}, most values in this range provide a satisfactory level of imbalance, so we pick a conservative value of $\nicefrac{1}{5n}$ as the default.
Albeit conservative, this value still results in a small cardinality for the head.
Even if keys follow a uniform distribution where all keys have probability $\nicefrac{1}{5n}$, the head will be comprised only by $5n$ keys.
In more skewed distributions, the size is substantially smaller.
Figure~\ref{fig:head-size} shows the cardinality of \hh for several Zipf distributions, for the two extremes of the range of $\theta$.
\begin{figure}[t]
\begin{center}
\includegraphics[width=0.8\columnwidth, height=4.5cm]{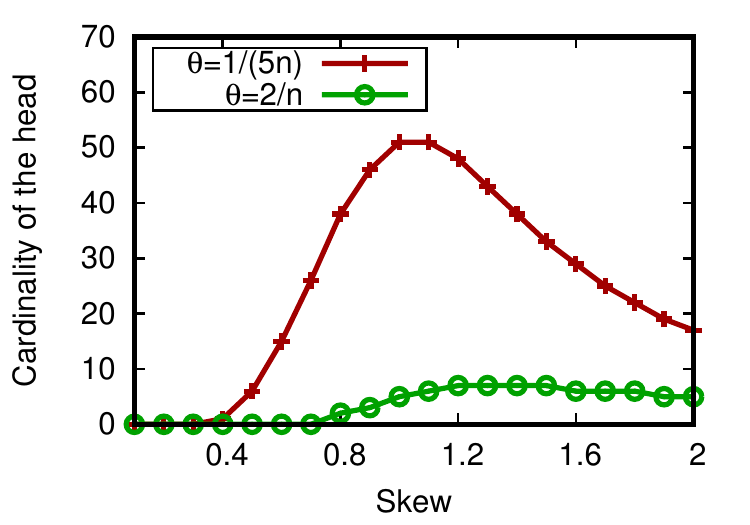}
\caption{Number of keys in the head of the distribution for $50$ and $100$ workers for a Zipf distribution with $z$ = $0.1 \ldots 2.0$, \keysize = $10^4$ and \nummsgs = $10^7$.}
\label{fig:head-size}
\end{center}
\vspace{-2em}
\end{figure}

\enlargethispage{0.5\baselineskip}
These numbers show that we we need to pay the price of added replication only for a small number of keys.
In order to make our proposal work completely online, we use a streaming algorithm to track the head and its frequency at runtime. 
In particular, we use the SpaceSaving algorithm~\cite{metwally2005spacesaving}, a counter-based algorithm for heavy hitters which can be extended to work in a distributed setting~\cite{berinde2010space}.

%

\subsection{Number of Choices for each Head Key}
Recall that we want to manage the head of the distribution differently from the tail, in order to achieve better load balance across workers.
In our proposed solution, both the head and tail make use of a \code{Greedy-$d$} process, but apply different $d$ per key.
Next we outline the basic principle behind this process.

\spara{Greedy-$d$.}
Let  $\hash{1}, \ldots, \hash{d}$ be $d$ independent hash functions that map $\keyspace \to [\numworkers]$ uniformly at random.
We define the \code{Greedy-$d$} process as follows: at time $t$, the $t$-th message (whose key is $k_t$) is placed on the worker with minimum current load among the candidate ones $\hash{1}(k_t), \ldots, \hash{d}(k_t)$, i.e., $P_t(k_t) = \argmin_{i \in \{\hash{1}(k_t),\ldots,\hash{d}(k_t)\}} L_i(t)$, where ties are broken arbitrarily.

Our method leverages this process to map keys to workers by hashing a key with a set of $d$ hash functions, and assigning the key to the least loaded of the $d$ available options.
In our solution, the tail has $d = 2$ choices as in \pkgs. 
However, the head can have $d \geq 2$, depending on the frequency of the keys.

We propose two different algorithms, depending on the value of $d$. 
These algorithms differ in their memory and aggregation overhead.

\begin{squishlist}
\item \textbf{\textit{D-Choices:}}
This algorithm adapts to the particular distribution of frequencies in the head, and uses a Greedy-$d$ process to handle the head, with $d < \numworkers$.
We show analytically how to pick $d$ in the following section.
\item \textbf{\textit{W-Choices:}}
This algorithm allows $\numworkers=|\workers|$ workers for the head. 
Conceptually, it is equivalent to setting $d \gg \numworkers \ln \numworkers$ for a Greedy-$d$ scheme, since this condition guarantees that there is at least one item in each bucket when using an ideal hash function.
In practice, there is no need to hash the keys in the head, and the algorithm can simply pick the least loaded among all workers.
\end{squishlist}

The pseudocode of the algorithm used by senders to determine the destination worker for their output messages is reported in Algorithm~\ref{alg:find-worker}.
\textsc{UpdateSpaceSaving} invokes the heavy hitters algorithm of~\cite{metwally2005spacesaving} and returns the set of the current heavy hitters. 
\textsc{FindOptimalChoices} determines the number of choices for heavy hitter keys as explained in Section~\ref{sec:analysis-d}. 
\textsc{MinLoad} selects the id of the worker with lowest load among the ones passed as arguments.
The load is determined based only on local information available at the sender, as in~\cite{nasir2015power}.
The algorithm uses $d$ different hash functions $\hash{i}$.

We also test our algorithms against a simple but competitive baseline, which has the same overhead as W-Choices:
\begin{squishlist}
\item \textbf{\textit{Round-Robin:}}
This algorithm performs round-robin assignment for the head across all available $\numworkers$ workers.
In comparison to W-Choices, this algorithm assigns the head to workers in a load-oblivious manner.
\end{squishlist}


\begin{algorithm}
\small
\caption{Stream partitioning algorithm.}
\pdfoutput=1

\DontPrintSemicolon
\SetKwProg{myfunct}{function}{}{}
\SetKw{Event}{upon}
\SetKwBlock{Do}{}{}

\Event message $m = \tuple{k, v}$ \Do{
	$\hh \leftarrow$ \textsc{UpdateSpaceSaving}$(k)$\;
	$d \leftarrow 2$ \tcp{Default as in \pkgs}
	\If{$k \in \hh$}{
		\If{\em \texttt{D-CHOICES}}{
			$d \leftarrow $ \textsc{FindOptimalChoices}()\;
		}
		\ElseIf{\em \texttt{W-CHOICES}}{
			$d \leftarrow \numworkers$\;
		}
	}

	$w \leftarrow$ \textsc{MinLoad}($\hash{1}(k),\ldots,\hash{d}(k)$)\;
	send$(w,m)$\;
}

\label{alg:find-worker}
\end{algorithm}

\section{Analysis}
\label{sec:analysis}
In this section, we analyze the D-Choices algorithm in order to find the number of choices $d$ that enables achieving load balance.
In addition, we discuss the memory requirement and overhead for the different algorithms.

\subsection{Setting $d$ for D-Choices}
\label{sec:analysis-d}
We first determine how to express $d$ analytically and then discuss how the function \textsc{FindOptimalChoices} computes a solution for that expression.

For this analysis, we assume that the key distribution \keydistr and the threshold $\theta$ are given.
Recall that we aim at minimizing $d$ in order to reduce the memory and aggregation overheads created by using a large number of choices, which are roughly proportional to $d \times \lvert \hh \rvert$.
Therefore, we would like to know which is the smallest number of choices $d$ that can lead to imbalance below some threshold $\epsilon$.

More formally, our problem can be formulated as the following minimization problem:
\begin{align}
\label{def:imbalance-dchoices}
\operatorname*{minimize}_d & \quad f(d ; \keydistr, \theta) = d \times \lvert \hh_{\keydistr, \theta} \rvert \,, \\
\label{def:imbalance-dchoices-constraint}
\operatorname*{subject\ to} & \quad \expect_d \left[ I(m) \right] \leq \epsilon \,.
\end{align}
In the problem formulation, we emphasize that \hh is a function of \keydistr and $\theta$.
Because both $\keydistr$ and $\theta$ are given, the objective function of the problem reduces to simply minimizing $d$.
Additionally, we employ a constraint on the expected imbalance generated when using the Greedy-$d$ process for \hh.

In order to attack the problem, we would like to express the imbalance $I(m)$ as a function of $d$.
However, the load on each worker $L_i$, from which the imbalance is computed, is a random variable that depends on the key assignment and the specific distribution, both of which are random.
The Greedy-$d$ process adapts dynamically to the load conditions, so it is hard to express the load analytically as a function of $d$.
Instead, we perform a \emph{lower bound} analysis on the expected load, and find analytically what are the \emph{necessary} conditions to obtain a feasible solution.
We then show empirically in Section~\ref{sec:evaluation} that the optimal values for $d$ are very close to this lower bound.

\begin{proposition}
\label{thm:main}
If the constraint of the optimization problem in Eqn.~\eqref{def:imbalance-dchoices-constraint} holds, then:
\begin{align}
\label{eq:final-prefix-bound}
\sum_{i \leq h} p_i + 
\left( \frac{ b_h }{\numworkers} \right)^{d} \sum_{h < i \leq \lvert \hh \rvert} p_i + 
\left( \frac{ b_h }{\numworkers} \right)^{2} \sum_{i > \lvert \hh \rvert} p_i 
\leq b_h \left( \frac{1}{\numworkers} + \epsilon \right), \\
\text{where } b_h = 
\numworkers - \numworkers \left( \frac{\numworkers - 1}{\numworkers} \right)^{h \times d}, \quad \forall k_h \in \hh. \notag
\end{align}
\end{proposition}

\begin{proof}
Let $L_w$ be the load of worker $w$.
By definition, $\avg_{w}(L_w)= \nicefrac{1}{\numworkers}$ since we consider normalized load that sums to $1$ over all workers.
Since $I(m) = \max_{w}(L_w) - \avg_{w}(L_w)$ by definition, and therefore $L_w - \nicefrac{1}{\numworkers} \leq I(m)$, we can rewrite Eqn.~\eqref{def:imbalance-dchoices-constraint} as:
\begin{align}
\label{eq:constraint-bound}
L_w - \frac{1}{\numworkers} \leq \epsilon, \quad \forall w \in \workers.
\end{align}

To continue our analysis, we would like to express the load on each worker as a function of the key distribution, however the load is a random variable that depends on the mix of keys that end up on the given worker, and the dynamic conditions of the other workers.
Nevertheless, while we cannot say anything about the load of a single worker, we can lower bound the (expected) load of a set of workers that handle a given key.

Let $\workers_i = \left\{ \hash{1}(k_i), \ldots, \hash{d}(k_i) \right\}$ be the set of workers that handle $k_i$.
Because workers are chosen via hashing, some of the $d$ choices might collide, and thus it might be that $\vert \workers_i \vert \leq d$.
Let $b = \lvert \workers_i \rvert$ be the expected size of $\workers_i$.

We now derive a lower bound on the cumulative (asymptotic) load on this set of workers.
This load consists of $p_i$ (by definition), plus the load generated by any key $k_j$ whose choices \emph{completely collide} with $k_i$, i.e., $\workers_j \subseteq \workers_i$.
Note that keys that partially collide with $k_i$ may also add up to the load of the workers in $\workers_i$, but we can ignore this additional load when deriving our lower bound.
The completely colliding keys $k_j$ may come from the head or from the tail. 
We take into account these two cases separately and express a lower bound on the expected cumulative load on $\workers_i$ as:
\begin{align}
\label{eq:load-bound}
\expect \left[ \sum_{w \in \workers_{i}} L_{w} \right] \geq 
p_i + 
\left( \frac{b}{\numworkers} \right)^{d} \sum_{\substack{j \leq \lvert \hh \rvert \\ j \neq i}} p_j + 
\left( \frac{ b }{\numworkers} \right)^{2} \sum_{j > \lvert \hh \rvert} p_j .
\end{align}
This equation expresses the frequency of $k_i$ ($p_i$) plus the (expected) fraction of the distribution that \emph{completely} collide with $k_i$.
A key in the head has $d$ independent choices, while a key in the tail has only $2$ independent choices.
Each choice has a probability of colliding with $\workers_i$ of $\nicefrac{b}{\numworkers}$ where $b$ is the size of $\workers_i$, as we assume ideal hash functions. 

We now use this lower bound to express a necessary condition for Eqn.~\eqref{eq:constraint-bound}, and thus for Eqn.~\eqref{eq:final-prefix-bound}.
By summing $L_w$ for all $w \in \workers_i$, Eqn.~\eqref{eq:constraint-bound} implies the following condition:
\begin{align}
\label{eq:sum-simple}
\expect \left[ \sum_{w \in \workers_{i}} L_{w} \right] \leq \frac{b}{n} + b \cdot \epsilon.
\end{align}

A {\em necessary} condition for this inequality to hold can be obtained by using the lower bound of Eqn.~\eqref{eq:load-bound}, which results in the following expression:
\begin{align}
p_i + 
\left( \frac{b}{\numworkers} \right)^{d} \sum_{\substack{j \leq \lvert \hh \rvert \\ j \neq i}} p_j + 
\left( \frac{ b }{\numworkers} \right)^{2} \sum_{j > \lvert \hh \rvert} p_j
\leq \frac{b}{\numworkers} + b \cdot \epsilon.
\end{align}

%
%

We can generalize the previous condition by requiring it to hold on the set of workers handling any prefix $p_1, \ldots, p_h$ of $h$ keys of the head. 
This results in a larger number of necessary conditions and thus in a potentially tighter bound, depending on the distribution of keys.
Similarly to Eqn.~\eqref{eq:sum-simple}, we obtain a set of constraints, one for each prefix of length $h$:

\begin{align}
\expect \left[ \sum_{w \in \bigcup_{i \leq h}\workers_{i}} L_{w} \right] \leq \frac{b_h}{\numworkers} + b_h \cdot \epsilon, \quad \forall k_h \in \hh,
\end{align}
where $b_h = \expect [ \lvert \bigcup_{h \leq i} \workers_{h} \rvert ]$ is the expected number of workers assigned to a prefix of the head of length $h$.
We can use once again the lower bound of Eqn.~\eqref{eq:load-bound} to obtain a necessary condition for the previous expression to hold as follows:
\begin{align}
\label{eq:constraint-prefix}
\sum_{i \leq h} p_i + 
\left( \frac{ b_h }{\numworkers} \right)^{d} \sum_{h < i \leq \lvert \hh \rvert} p_i + 
\left( \frac{ b_h }{\numworkers} \right)^{2} \sum_{i > \lvert \hh \rvert} p_i \leq &\\
\leq 
\expect \left[ \sum_{w \in \bigcup_{i \leq h}\workers_{i}} L_{w} \right] \leq \frac{ b_h }{\numworkers} + b_h \cdot \epsilon, & \quad \forall k_h \in \hh. \notag
\end{align}

We derive $b_h$ in Appendix~\ref{sec:collisions} as follows:
\begin{align}
b_h = \expect \left[ \left\vert \bigcup_{i \leq h} \workers_{i} \right\vert \right] = \numworkers - \numworkers \left( \frac{\numworkers - 1}{\numworkers} \right)^{h \times d}.
\end{align}


%
%

Finally, after substituting $b_h$ we obtain the final set of bounds of Eqn.~\eqref{eq:final-prefix-bound}, one for each prefix of the head \hh.
\end{proof}

Solving analytically for $d$ is complex, so we use an alternative approach to implement \textsc{FindOptimalChoices}.
We simply start from $d=p_1 \times n$, which is a simple lower bound on $d$ that derives from the fact that we need $p_1 \leq \nicefrac{d}{n}$, and increase $d$ until all constraints are satisfied.
In practice, the tight constraints are the ones for $h=1$, the most frequent key, and for $h=|\hh|$, when considering the whole head.
The latter, is especially tight for very skewed distributions, where the head represents a large fraction of the total load.
In this case, we need $b_h \approx \numworkers$ with high probability, but, as well know, this happens only when $h \times d \gg \numworkers \ln \numworkers$.
Given that $\max h$ is small, as the distribution is skewed and \hh has small cardinality, we need a large $d$.
However, having $d \geq \numworkers$ is not sensible, so when this condition is reached, the system can simply switch to the W-Choices algorithm.

\subsection{Memory Overhead}
\label{sec:memory}

We now turn our attention to the cost of the algorithms under study, as measured by the memory overhead required to achieve load balance.
Henceforth, we refer mainly to the memory overhead, however note that when splitting a key in $d$ separate partial states, if reconciliation is needed, there is also an aggregation cost proportional to $d$.
The two costs are proportional, so we focus our attention only on the first one.

\spara{Overhead on Sources.}
All three algorithms have the same cost on the sources.
The sources are responsible for forwarding messages to the workers.
Each source maintains a load vector storing local estimations of the load of every worker.
This estimation is based only on the load each source sends, and it is is typically a very accurate approximation of the actual global load of workers obtained considering all sources~\cite{nasir2015power}.
Storing the load vector requires $\mathcal{O}(\numworkers)$ memory.
Thus, the total memory required for estimating the load of workers at each source is $\numsources \times \numworkers$.
These costs are the same for all algorithms, even though Round-Robin uses the load information only for the tail.
Alongside the load estimation, each source runs an instance of the SpaceSaving algorithm to estimate the frequency of the head keys in the input streams.
The SpaceSaving algorithm requires $\mathcal{O}(1)$ memory and $\mathcal{O}(1)$ processing time per message.

\spara{Overhead on Workers.}
Without loss of generality, let us assume the memory required to maintain the state for each key is unitary.
Recall the set of choices available per key:\footnote{Ignoring that some keys in the tail have frequency smaller than \numworkers.}
\begin{squishlist} 
\item For $d = 2$, each key is assigned to at most two workers, as in \pkgs. The upper bound on the memory requirement to store \keysize unique items is $2 \times \keysize$.
\item For $2 < d < \numworkers$, each key is assigned to at most $d$ workers, as in D-Choices. The upper bound on the memory requirement is $d \times \lvert \hh \rvert + 2 \times \lvert \tail \rvert $.
\item For $d \gg \numworkers \ln \numworkers$, each key is assigned to all the \numworkers workers, as in W-Choices and Round-Robin. The upper bound on the memory requirement is $\numworkers \times \lvert \hh \rvert + 2 \times \lvert \tail \rvert$.
\end{squishlist}




From these formulas, it is clear that D-Choices always has a smaller overhead than the other algorithms, although it is larger than \pkgs.
But how much is the overhead in theory?
By using Inequality~\eqref{eq:final-prefix-bound}, we can find the appropriate value for $d$, in case the distribution is known apriori.
Figure~\ref{fig:bounds_d1_10000} shows the overhead of D-Choices for the required $d$, as a fraction of the number of workers, as a function of the skew.
We vary the exponent of a Zipf distribution $z \in \{0.1,\ldots, 2.0\}$, with $\keysize = 10^4$ unique keys, and $m = 10^7$ messages. 
The plot shows that D-Choices reduces the number of workers compared to W-Choices and Round-Robin.
Especially at larger scales, $\numworkers=50$ and $\numworkers=100$ workers, $d$ is always smaller than $n$.

\begin{figure}[t]
\begin{center}
\includegraphics[width=0.8\columnwidth, height=4.5cm]{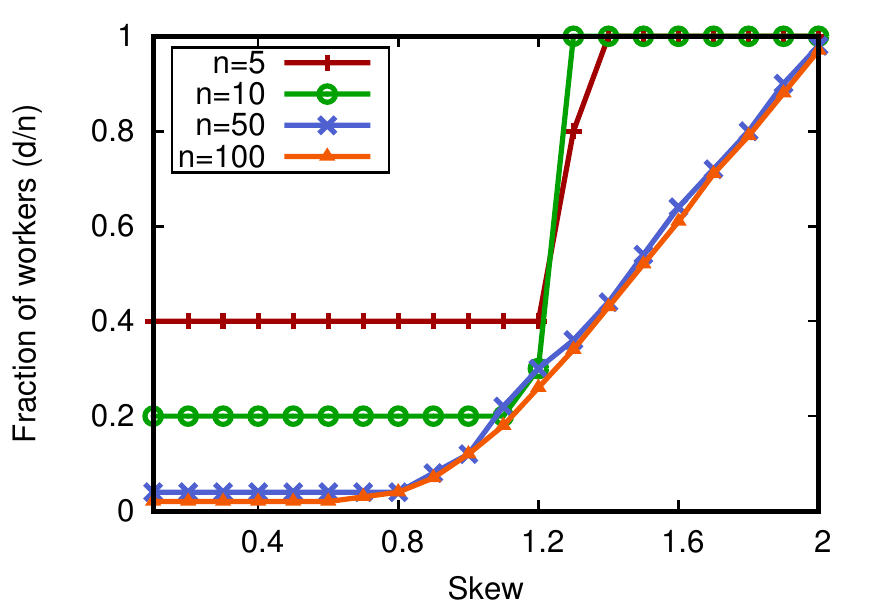}
\caption{Fraction of workers used by D-Choices for the head as a function of skew $z \in \{0.1,\ldots, 2.0\}$. Values for $\keysize=10^4$, $m = 10^7$, and $\epsilon = 10^{-4}$. }
\label{fig:bounds_d1_10000}
\end{center}
\vspace{-2em}
\end{figure}

%

Lastly, Figure~\ref{fig:zipf_memory_pkg} and Figure~\ref{fig:zipf_memory_sg}  show the estimated memory used by D-Choices and W-Choices with respect to \pkgs and \sg respectively.
All the algorithms (including \pkgs) need an aggregation phase, so the comparison is fair.
We estimate the memory for \pkgs and \sg using $\operatorname{mem}_{\pkgs} = \sum_{k \in \keyspace} \min(f_k, 2)$ and $\operatorname{mem}_{\sg} = \sum_{k \in \keyspace} \min(f_k,  \numworkers)$ respectively.

We derive three main observations from the results of Figure~\ref{fig:zipf_memory_pkg} and Figure ~\ref{fig:zipf_memory_sg}.
First, both D-Choices and W-Choices require an amount of memory similar to \pkgs when the number of workers is small.
However, this behavior changes when the number of workers increases.
For instance, for 100 workers W-Choices requires up to 25\% more memory compared to \pkgs.
Second, D-Choices requires less memory than W-Choices when the skew is moderately high (note that ratios lower than $10^{-2}$ are not plotted).
Nevertheless, both curves approach the same memory overhead at very high skew.

Note that the bump in Fig.~\ref{fig:zipf_memory_pkg} primarily depends on the size of the head (as shown in Fig~\ref{fig:head-size}) and how a Zipf distribution with finite support is defined.
Intuitively, as the skew grows larger, initially more keys pass the frequency threshold $\theta$, but later only a few keys become dominant (to the limit, only one key is in the head with $p_1 \rightarrow 1$).

Third, the overhead compared to \sg is negligible, especially at larger scales.
Both D-Choices and W-Choices require only a fraction of the cost of shuffle grouping, fulfilling one of the two desiderata for our solution.
The next section shows both proposed algorithms also achieve low imbalance, which is their main purpose.

\begin{figure}[t]
\begin{center}
\includegraphics[width=\columnwidth]{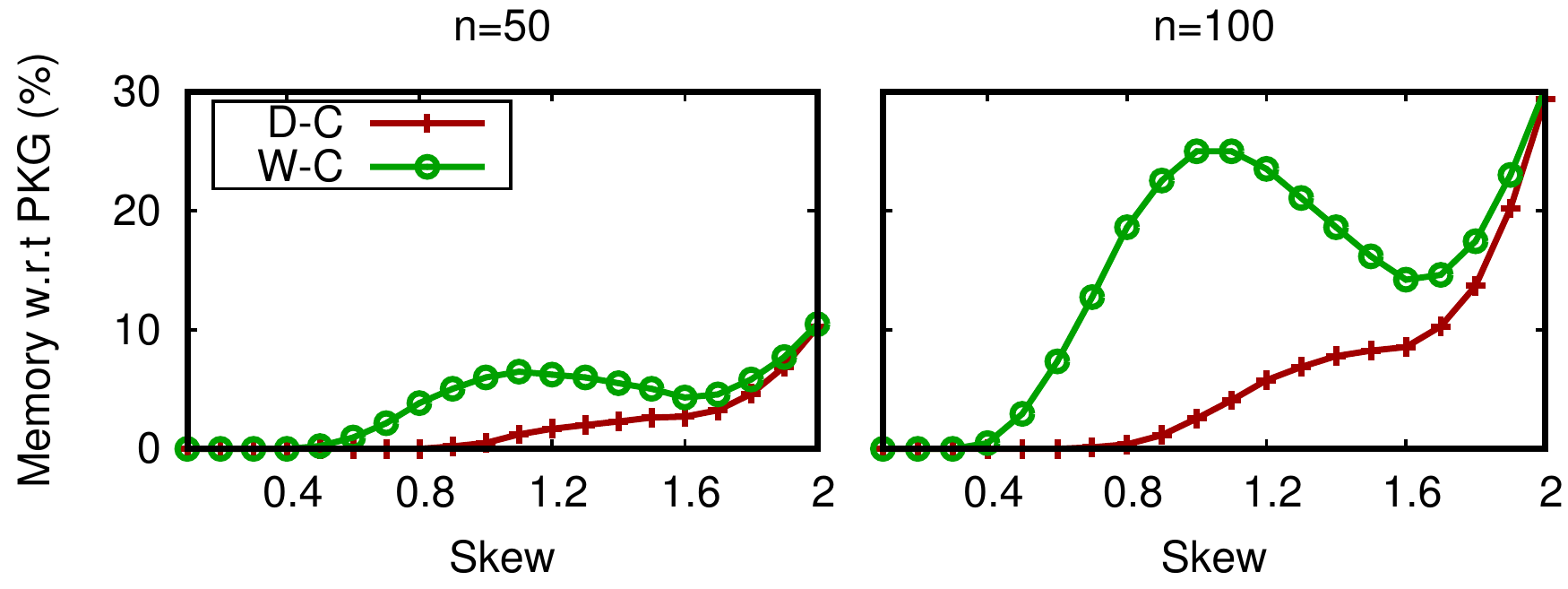}
\caption{Memory overhead for D-Choices (\dc) and W-Choices (\wc) with respect to \pkgs as a function of skew, for different number of workers $n \in \{50,100\}$. Values for $\keysize = 10^4$ and $\epsilon = 10^{-4}$. In the worst case, \dc and \wc use at most $30\%$ more memory than \pkgs, and \dc uses considerably less memory than \wc in most cases.}
\label{fig:zipf_memory_pkg}
\vspace{-1.5em}
\end{center}
\end{figure}
\begin{figure}[t]
\begin{center}
\includegraphics[width=\columnwidth]{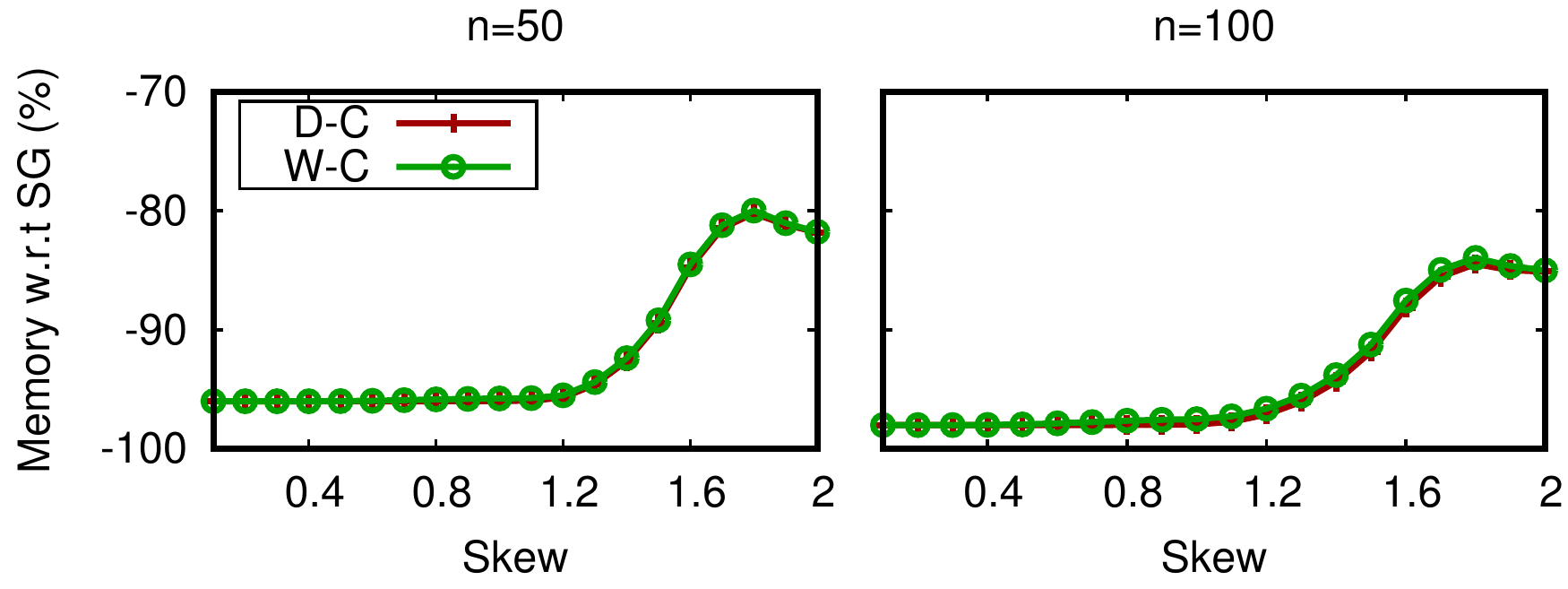}
\caption{Memory overhead for D-Choices (\dc) and W-Choices (\wc) with respect to \sg as a function of skew, for different number of workers $n \in \{50,100\}$. Values for $\keysize = 10^4$ and $\epsilon = 10^{-4}$. In the worst case, \dc and \wc still use $80\%$ less memory than \sg.}
\label{fig:zipf_memory_sg}
\vspace{-1.5em}
\end{center}
\end{figure}

\section{Evaluation}
\label{sec:evaluation}

In Section~\ref{sec:memory} we have derived an analysis of the memory overhead of our proposed techniques.
We now assess their load balancing performance using both simulations and a real deployment.
In so doing, we answer the following questions:

\begin{squishlist}
\setlength{\labelwidth}{1.5em}
\item[\textbf{Q1:}]
How to decide the threshold $\theta$ that defines the head of the distribution?
\item[\textbf{Q2:}]
How close is the estimated value of $d$ for \dc to an empirical optimum?
\item[\textbf{Q3:}]
How do the proposed algorithms compare in terms of load imbalance?
\item[\textbf{Q4:}]
What is the overall effect of proposed method on the throughput and latency of a real \dspe?
\end{squishlist}

\subsection{Experimental Setup}
\label{sec:datasets}

\begin{table}[t]
\caption{Summary of the datasets used in the experiments: number of messages, number of keys, and probability of the most frequent key ($p_1$).} 
\centering
\small
\begin{tabular}{l c r r r}
\toprule
Dataset		&	Symbol	&	Messages				&	Keys			& 	$p_1$(\%)		\\
\midrule
Wikipedia		&	WP		&	\num{22}M		&	\num{2,9}M		&	\num{9,32}	\\
Twitter		&	TW		&	\num{1,2}G		&	\num{31}M		&	\num{2,67}	\\
Cashtags		&	CT		&	\num{690}k	&	\num{2,9}k	&	\num{3,29}	\\
\midrule
Zipf			&	ZF		&    $10^7$   	&  $10^4$,$10^5$,$10^6$ &   $\propto$ $\frac{1} { \sum{x^{-z}} }$ \\
\bottomrule
\end{tabular}
\label{tab:summary-datasets}
\end{table}

\begin{table}[t]
\caption{Notation for the algorithms.}
\centering
\small
\begin{tabular}{r l l}
\toprule
Symbol	&	Algorithm		&	Head vs. Tail \\
\midrule
\dc	& D-Choices	&	\multirow{3}{*}{Specialized on head} \\
\wc & W-Choices 	&	\\
\rr & Round-Robin	&	\\
\midrule
\pkgs & Partial Key Grouping		&	\multirow{2}{*}{Treats all keys equally}\\
\sg   & Shuffle Grouping & 	\\
\bottomrule
\end{tabular}
\label{tab:summary-algorithm}
\end{table}

\begin{table}[t]
\caption{Default parameters for the algorithms.}
\centering
\small
\begin{tabular}{c l c }
\toprule
Parameter		&	Description	&	Values \\
\midrule
\numworkers	&	Number of workers	&	5, 10, 20, 50, 100 \\
\numsources	&	Number of sources	&	$5$ \\
$\epsilon$		&	Imbalance tolerance (D-Choices)	&	$10^{-4}$ 	\\
$\theta$		&	Threshold defining the head	&	$\frac{2}{n}, \ldots, \frac{1}{8n}$ \\
\bottomrule
\end{tabular}
\label{tab:algorithmic-parameters}
\end{table}

\spara{Datasets.}
Table~\ref{tab:summary-datasets} summarizes the datasets used.
We use three real-world datasets, extracted from Wikipedia and Twitter.
The Wikipedia dataset (WP)\footnote{\url{http://www.wikibench.eu/?page\_id=60}} is a log of the pages visited during a day in January 2008~\citep{urdaneta09wikipedia}.
Each visit is a message and the page's URL represents its key.
The Twitter dataset (TW) is a sample of tweets crawled during July 2012.
Each tweet is split into words, which are used as the key for the message.
Additionally, we use a Twitter dataset that is comprised of tweets crawled in November 2013.
The keys for the messages are the \emph{cashtags} in these tweets.
A \emph{cashtag} is a ticker symbol used in the stock market to identify a publicly traded company preceded by the dollar sign (e.g., \$AAPL for Apple).
This dataset allows to study the effect of drift in the skew of the key distribution.
Finally, we generate synthetic datasets (ZF) with keys drawn from Zipf distributions with exponent in the range $z \in \{ 0.1, \ldots, 2.0 \}$, and with a number of unique keys in the range $\keysize \in \{ 10^{4}, 10^{5}, 10^{6} \}$.

\spara{Simulation.}
We process the datasets by simulating the simplest possible \dagr, comprising one set of sources (\sources), one set of workers (\workers), and one intermediate partitioned stream from sources to workers.
The \dagr reads the input stream from multiple independent sources via shuffle grouping.
The sources forward the received messages to the workers downstream via the intermediate stream, on which we apply our proposed grouping schemes.
Our setting assumes that the sources perform data extraction and transformation, while the workers perform data aggregation, which is the most computationally expensive part of the \dagr, and the focus of the load balancing.

\spara{Algorithms.}
Table~\ref{tab:summary-algorithm} defines the notations used for different algorithm, while Table~\ref{tab:algorithmic-parameters} defines values of different parameters that we use for the experiments.
For most of the experiments we use load imbalance as defined in Section~\ref{sec:preliminaries} as a metric to evaluate the performance of the algorithms.
Unlike the algorithms in Table~\ref{tab:summary-algorithm}, other related load balancing algorithms~\citep{shah2003flux,cherniack2003scalable,xing2005dynamic,balkesen2013adaptive,castro2013integrating} require the \dspe to support operator migration.
Many \dspes, such as Apache Storm, do not support migration, so we omit these algorithms from the evaluation.

\subsection{Experimental Results}

\begin{figure*}[t]
\begin{center}
\includegraphics[scale = 0.55]{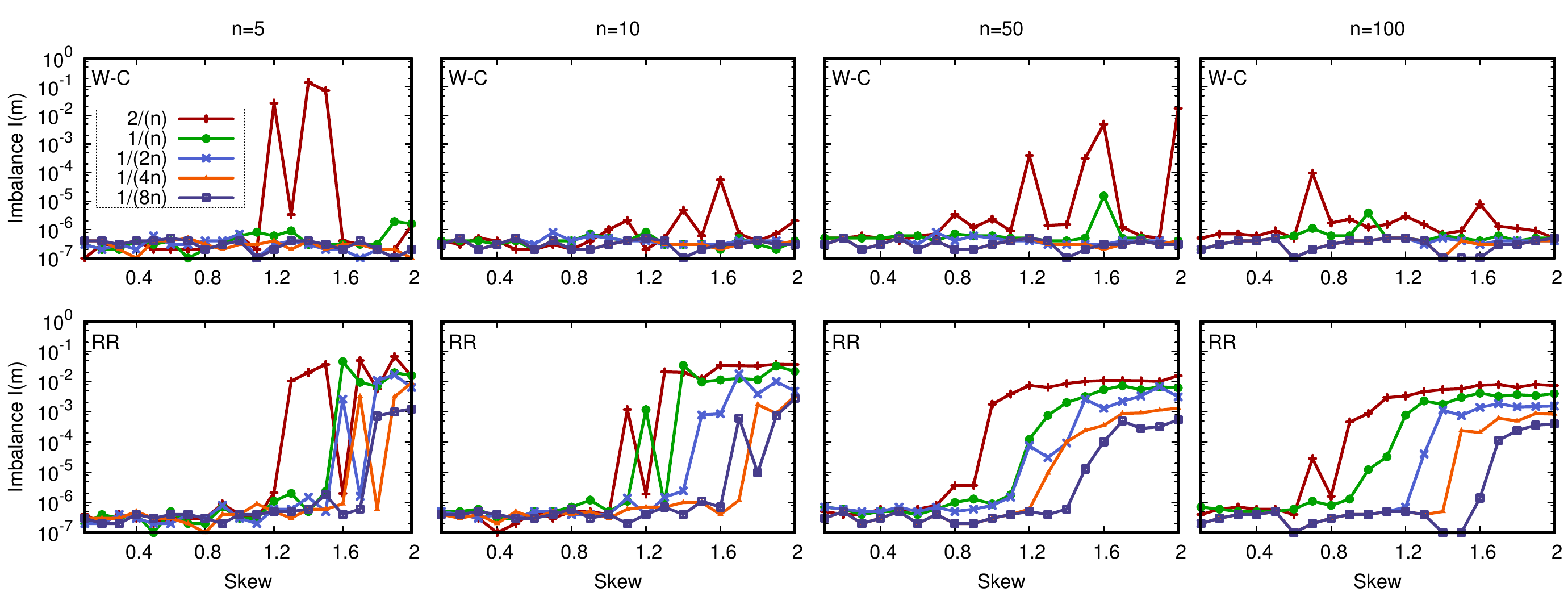}
\caption{Load imbalance as a function of skew $z$, for different thresholds $\theta$, for \wc and \rr. \wc provides a better load balance than \rr at high skew, especially for larger scale deployments, while having the same memory overhead. For this experiment, the stream has $\keysize = 10^4$ unique keys, and $m = 10^7$ messages.}
\label{fig:threshold-load-imbalance}
\end{center}
\vspace{-2em}
\end{figure*}
\begin{figure*}[t]
\begin{center}
\includegraphics[width=0.7\textwidth]{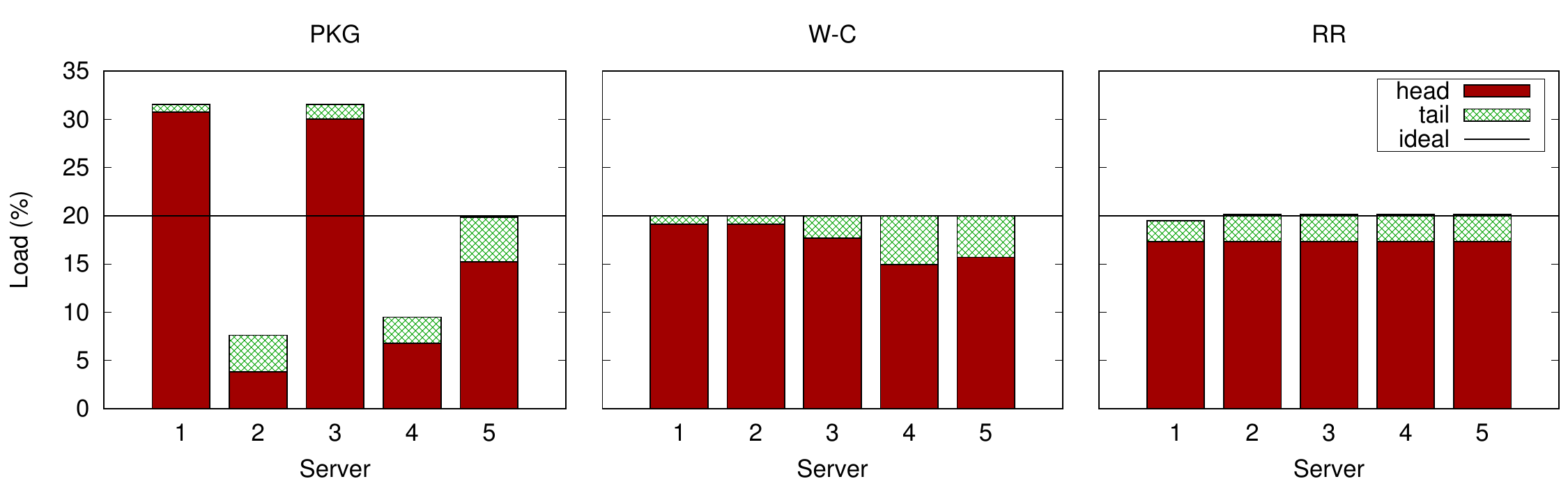}
\caption{Load generated by head and tail when using \pkgs, \wc , and \rr. The horizontal dotted line represents the ideal even distribution of the load (\nicefrac{1}{\numworkers}). For this experiment, the threshold is set to $\theta = \nicefrac{1}{8\numworkers}$, and the stream has $m=10^7$ messages with $\keysize = 10^4$ drawn from a Zipf distribution with exponent $z=2.0$.}
\label{fig:heavy_tail_2.0}
\end{center}
\vspace{-2em}
\end{figure*}
\begin{figure}[t]
\begin{center}
\includegraphics[width=0.9\columnwidth]{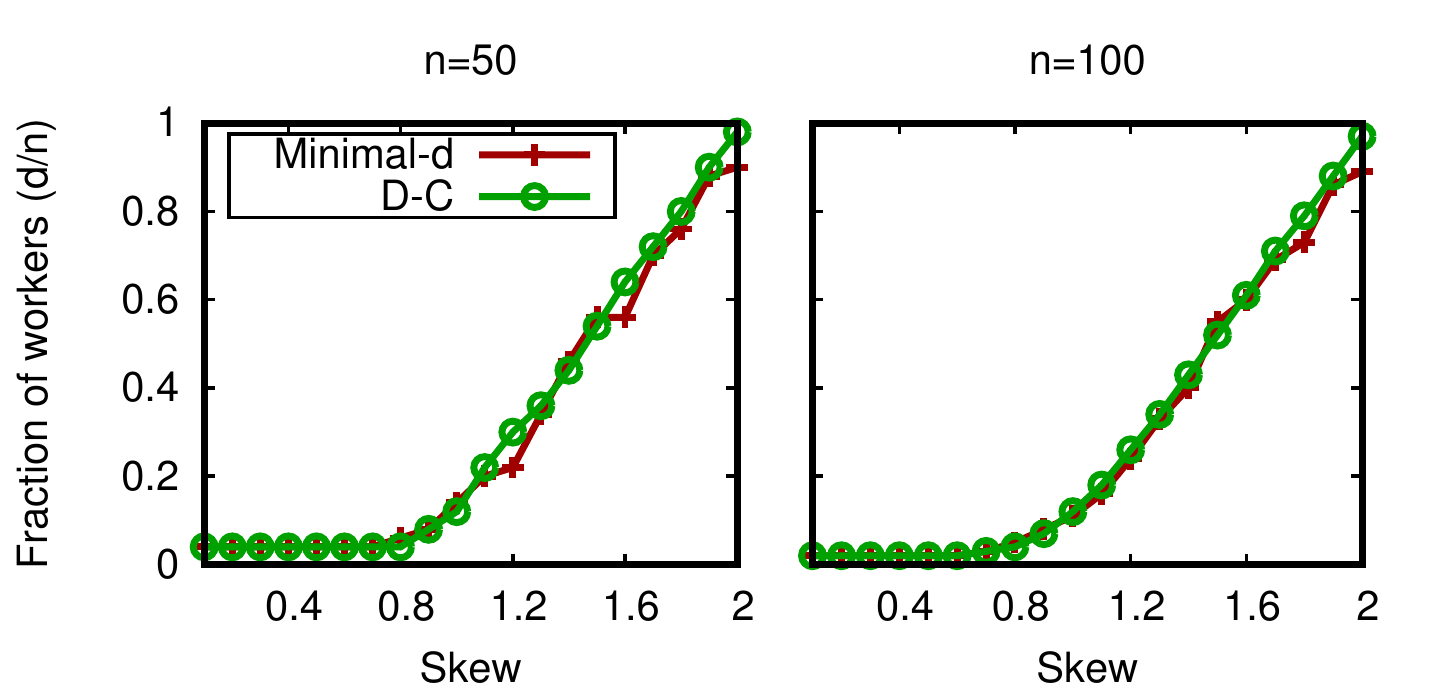}
\caption{Comparison of the $d$ value for \dc to the minimum value of $d$ required to match the imbalance of \wc (Min-d). The value computed by \dc is very close to the minimum for all settings.}
\label{fig:greedydvsdchoice}
\end{center}
\vspace{-2em}
\end{figure}
\begin{figure*}[t]
\begin{center}
\includegraphics[width=0.9\textwidth]{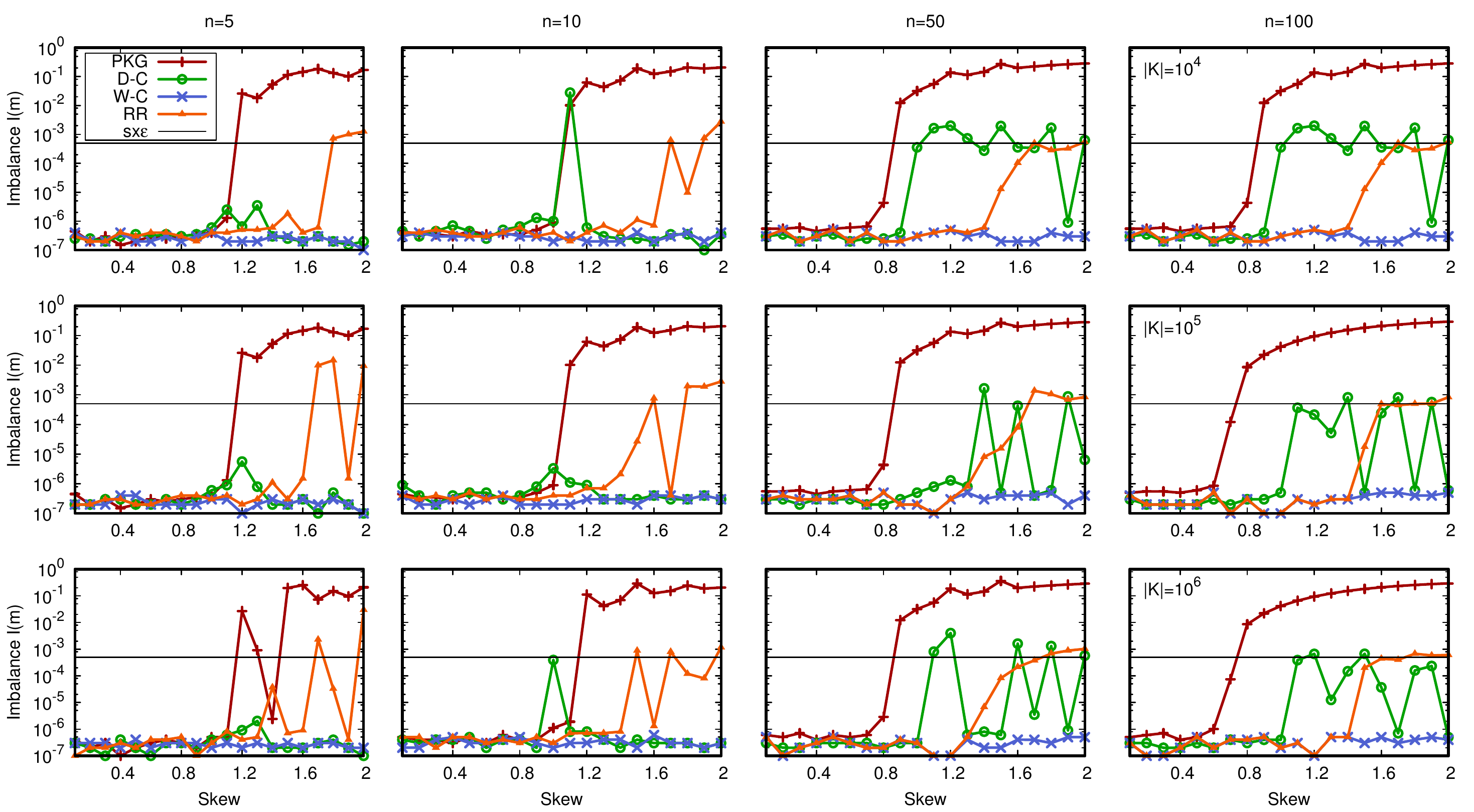}
\caption{Load imbalance on the ZF datasets for \pkgs, \dc, \wc, and \rr as a function of skew. The plots are shown for several combinations of settings of number of workers \numworkers and number of unique keys \keysize (number of messages $m = 10^7$).}
\label{fig:zipf_load_imbalance}
\end{center}
\vspace{-2em}
\end{figure*}
\spara{Q1.}
To find the right threshold $\theta$, we run an experiment with the W-Choices algorithm.
We use \wc because it is potentially the best one in terms of imbalance, as it is able to use the entire set of workers to distribute the load.
We compare it with Round-Robin, which has the same cost in terms of memory.
We vary the threshold by halving it, starting from $\nicefrac{2}{\numworkers}$ down to $\nicefrac{1}{8\numworkers}$.
For this experiment, we use the ZF dataset in order to verify the effect of skew.

Figure~\ref{fig:threshold-load-imbalance} shows the load imbalance as a function of skew, for different number of workers.
For both algorithms, reducing the threshold (i.e., considering a larger head) reduces the imbalance, and increasing the skew (i.e., larger $z$) increases the imbalance, as expected.
However, \wc achieves ideal load balance for any threshold $\theta \leq \nicefrac{1}{\numworkers}$, while \rr exhibits a larger gradient of results.
This behavior is more evident at larger scales, where \rr starts generating imbalance even under modest levels of skew.

The difference in results between the two algorithms is entirely due to the fact that \wc uses a load-sensitive algorithm to balance the head, which adapts by taking into consideration also the load of the tail.
Conversely, \rr balances the head and tail separately.
As a result, while the head is perfectly balanced, the tail has less flexibility in how to be balanced.

As additional evidence, Figure~\ref{fig:heavy_tail_2.0} shows the workload generated by head and tail separately for \pkgs, \wc, and \rr.
To ease visual inspection, we show the case for $\numworkers = 5$; the other cases follow a similar pattern.
This plot is obtained for $\theta = \nicefrac{1}{8\numworkers}$, i.e., the lowest threshold in our experiments.
As expected, the most frequent key in \pkgs overloads the two workers responsible for handling it, thus giving rise to high load imbalance.
Instead, \wc mixes the head and tail across all the workers, and achieves ideal load balance.
Finally, \rr splits the head evenly across the workers, but falls short of perfect balance, as represented by the horizontal dotted line.

In summary, even for low thresholds, \wc behaves better than \rr, and a threshold of $\nicefrac{1}{\numworkers}$ is enough for \wc to achieve low imbalance in any setting of skew and scale.
In the following experiments we set the threshold to its default value $\theta = \nicefrac{1}{5n}$.
This conservative setting gives the algorithms the largest amount of freedom in setting $d$, and guarantees that any imbalance is due to the choice of $d$ rather than the choice of $\theta$, which determines the cardinality of \hh.



\spara{Q2.}
In this experiment, we evaluate the parameter $d$ for \dc.
We compare how far the $d$ value as computed by the \dc algorithm is from an empirically-determined optimal value.
To find the optimum, we execute the Greedy-$d$ algorithm with all the possible values of $d$, i.e., $2, \ldots, n$ and find the minimum $d$ that is required to match the imbalance achieved by \wc.

Figure~\ref{fig:greedydvsdchoice} shows this optimum in comparison to the value of $d$ for \dc as computed by the algorithm described in Section~\ref{sec:analysis}.
The results clearly show that \dc is very close to the optimal $d$, thus supporting our analysis.
In all cases, \dc uses a $d$ slightly larger than the minimum, thus providing good load balance at low cost.

\begin{figure*}[t]
\begin{center}
\includegraphics[scale=0.65]{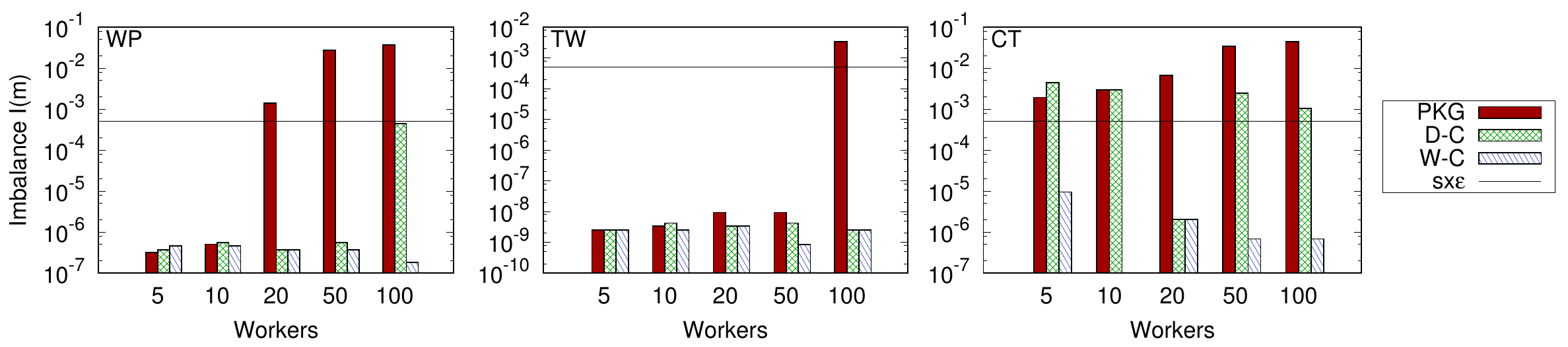}
\caption{Imbalance on the real-world datasets WP, TW and CT for \pkgs , \dc, and \wc as a function of the number of workers. Larger scale deployments are inherently harder to balance.}
\label{fig:load-imbalance-datasets}
\end{center}
\vspace{-2em}
\end{figure*}

\begin{figure*}[t]
\begin{center}
\includegraphics[width=\textwidth]{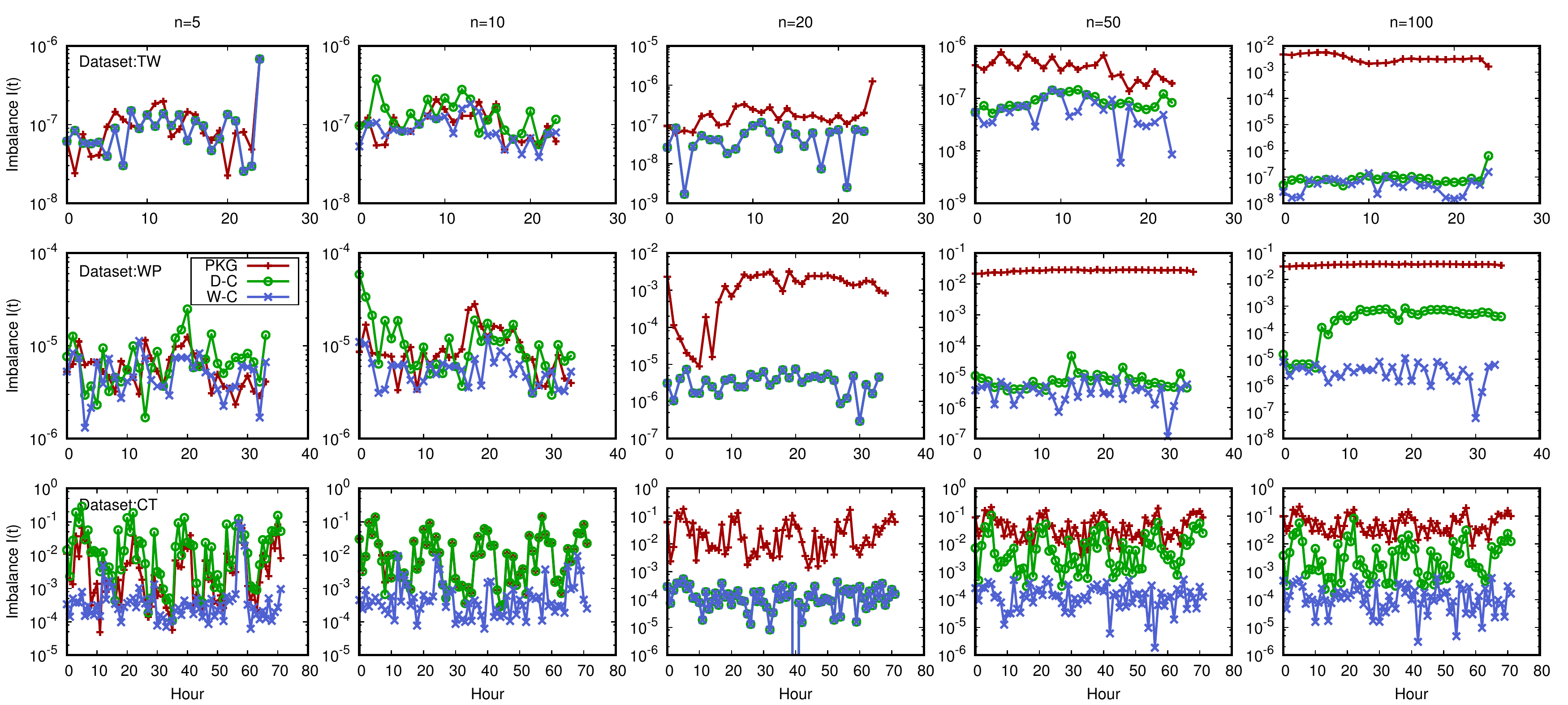}
\caption{Load imbalance over time for the real-world datasets TW, WP, and CT.}
\label{fig: time_series_twitter_wiki_tickers}
\end{center}
\vspace{-2em}
\end{figure*}

\spara{Q3.}
We now examine the performance of the proposed algorithms in terms of load imbalance.
We use the ZF datasets to control for the skew, and additionally, we evaluate the proposed algorithms on several real-world datasets, i.e., WP, TW and CT. 
We compare \dc and \wc to \rr and \pkgs as baselines.

Figure~\ref{fig:zipf_load_imbalance} shows the average imbalance as a function of skew for several possible settings of number of workers and number of keys.
The number of keys does not influence the imbalance, thus suggesting that $\keysize = 10^4$ is a large enough number to give the algorithms enough freedom in their placement, for these particular distributions.
On the other hand, skew and scale clearly play an important role.
The problem becomes harder as both $z$ and $n$ increase.
\wc is the best performer in terms of imbalance, as it manages to keep it constantly low irrespective of the setting.
\dc and \rr perform similarly, and both manage to keep the imbalance much lower than \pkgs for larger values of skew and scale.
However, recall that \dc does so at a fraction of the cost of \rr.

Figure~\ref{fig:load-imbalance-datasets} shows similar results for the real-world datasets (we omit \rr as it is superseded by \wc).
All algorithms perform equally well when the number of workers is small.
Instead, at larger scales ($\numworkers=20, 50, 100$) \pkgs generates higher load imbalance compared to \dc and \wc.
As expected, \wc provides better load balance compared to \dc .
However, recall that for \dc the level of imbalance is also determined by the tolerance $\epsilon$.
As each source executes the algorithm independently, in the worst case we should expect an imbalance of $s \times \epsilon$, plotted as a horizontal dotted line.

The Twitter cashtags dataset (CT) behaves differently from TW and WP.
This dataset is characterized by high concept drift, that is, the distribution of keys changes drastically throughout time.
This continuous change poses additional challenges to our method, especially for the heavy hitters algorithm that tracks the head of the distribution.
Overall, the dataset is harder to balance for all methods, however the relative performance of the various algorithms is still as expected.
\dc works better than \pkgs at larger scales, and \wc provides even better load balance.

Finally, Figure~\ref{fig: time_series_twitter_wiki_tickers} allows to observe the behavior of the proposed algorithms over time.
This experiment enables us to study the effect of drift in the key distribution.
We compare the imbalance over time for \pkgs, \dc, and \wc on the real-world datasets.
The plot shows the imbalance for several settings, and confirms our previous observations.
Larger scales are challenging for \pkgs, but are handled well by \dc and \wc, and the concept drift in the CT dataset makes the task harder.
Nevertheless, overall the imbalance remains relatively stable throughout the simulation.

%

\begin{figure}[h!]
\begin{center}
\includegraphics[width=0.9\columnwidth, height=3.7cm]{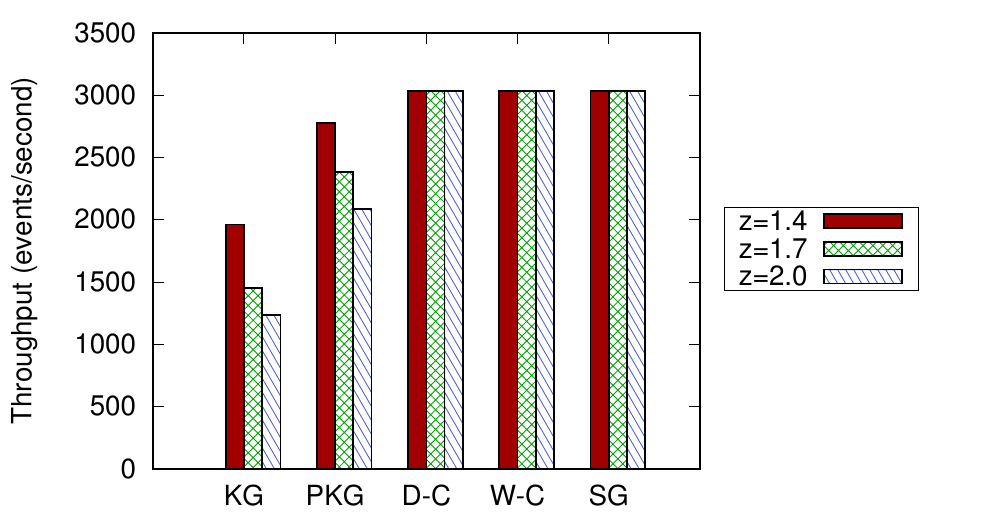}
\caption{Throughput on a cluster deployment on Apache Storm for \kg, \pkgs , \sg, \dc, and \wc on the ZF dataset, for sample values of $z$, with $\numworkers = 80$, \keysize =$10^4$, and \nummsgs = $2 \times 10^6$.}
\label{fig:throughput}
\end{center}
\vspace{-2em}
\end{figure}
\begin{figure*}[t]
\begin{center}
\includegraphics[width=0.9\textwidth]{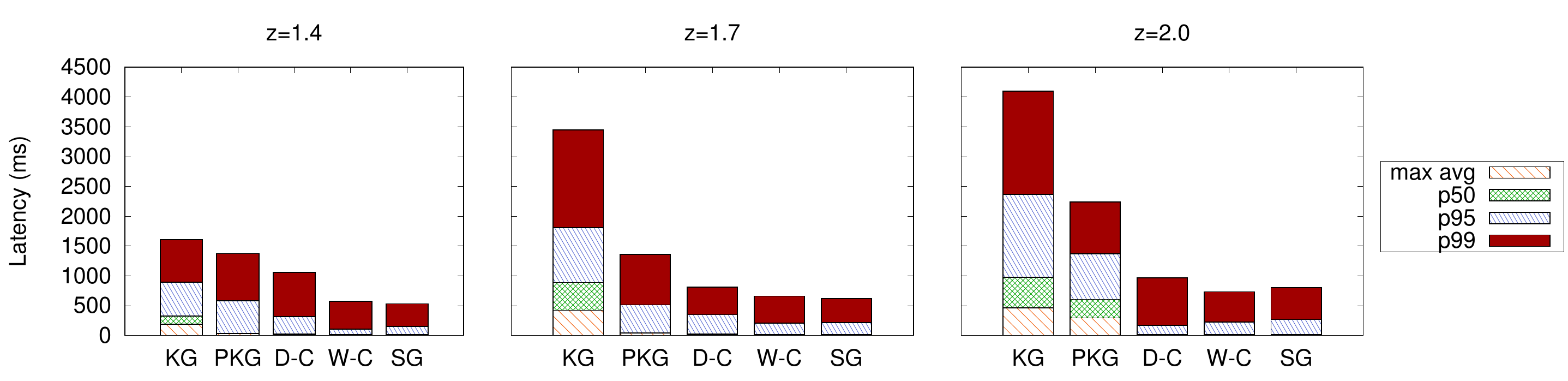}
\caption{Latency (average and percentiles) on a cluster deployment on Apache Storm for \kg, \pkgs , \sg, \dc, and \wc on the ZF datasets, for sample values of $z$, with $\numworkers = 80$, $\keysize =10^4$, and $\nummsgs = 2 \times 10^6$.}
\label{fig:latency-2}
\end{center}
\vspace{-2em}
\end{figure*}

\enlargethispage{\baselineskip}
\spara{Q4.}
To quantify the impact of our proposed techniques, \dc and \wc, on a real \dspe, we implement them in Apache Storm and deploy them on a cluster.
We compare them with the standard groupings: \kg, \sg, and \pkgs.
The Storm cluster consists of \num{9} machines (1 master and 8 slaves), where each machine has $16$ VCPUs.
Each slave machine has \num{16} execution slots (for a total of $8 \times 16 = 128$ slots in the cluster).

We use a simple streaming application that consist of two operators: sources that generate the input stream, and workers that perform aggregation on the stream.
Concretely, we use $48$ sources and $80$ workers.
The sources generate the stream from a Zipf distribution with skew $z \in \{1.4, 1.7, 2.0\}$, \keysize = $10^4$, and \nummsgs = $2 \times 10^6$.
This aggregation-like application is representative of many data mining algorithms, such as computing statistics for classification, or extracting frequent patterns.
To emulate some CPU consumption, we add a fixed delay of $1$ms to the processing of each message (this delay represents $\nicefrac{1}{10}$-th of a disk seek\footnote{\url{http://brenocon.com/dean_perf.html}}).
We choose this value for the delay to bring the application to the saturation point of the cluster.
We report overall throughput and end-to-end latency by taking the average over multiple iterations.

\enlargethispage{\baselineskip}
Figure~\ref{fig:throughput} shows the throughput for the experiment.
As expected, \kg achieves the lowest throughput among all the groupings.
\pkgs performs better than \kg, however it is not able to match \sg, and pays a price for the imbalance.
Instead, our proposed techniques (\dc and \wc) are able to match the throughput of \sg (at a fraction of its cost).
As expected, increasing the skew makes the problem harder for both \kg and \pkgs.
In the best case, \dc and \wc achieve a throughput $1.5$ times higher than \pkgs, and $2.3$ times higher than \kg.

Finally, Figure~\ref{fig:latency-2} shows the difference among the algorithms in terms of latency.
The plot reports the maximum of the per-worker average latencies, along with the $50$th, $95$th, and $99$th percentiles across all workers.
This message latency primarily depends on the time each message waits in the input queue of the worker before being processed, since we keep computation time per message constant to $1$ms.

While with modest skew the difference is small, for larger skews \kg suffers from very high latency.
Given that the worker handling the most frequent key needs to process a large fraction of the stream, longer queues form at this worker, and this result is to be expected.
Clearly \pkgs fares better, introducing about $50\%$ of the latency of \kg at the highest skew.
However, for all settings it is not able to match the latency provided by \sg, which can be up to $3$ times smaller.
Conversely, both \dc and \wc are able to provide latencies very close to the ones obtained by \sg.
In the best case, \dc and \wc reduce the $99$th percentile latency of \pkgs by $60\%$, and the latency of \kg by over $75\%$.

Overall, both D-Choices and W-Choices are able to achieve the same throughput and latency of shuffle grouping at just a fraction of the cost in terms of memory.
These results suggest that both algorithms could easily replace the default implementation of shuffle grouping in most \dspes, and provide a form of ``worker affinity'' for states related to a given key without compromising on the load balance.


\section{Related Work}
\label{sec:rel-work}
Recently, there has been considerable interest in the load balancing problem for \dspes when faced with a skewed input stream~\cite{gedik2014partitioning, rivetti2015efficient}.
Similarly to our work, previous solutions employ heavy hitter algorithms, and provide load balancing for skewed inputs when partitioned on keys. 
\citet{gedik2014partitioning} develops a partitioning function (a hybrid between explicit mapping and consistent hashing) for stateful data parallelism in \dspes that leverages item frequencies to control migration cost and imbalance in the system.
The author proposes heuristics to discover new placement of keys on addition of parallel channels in the cluster, while minimizing the load imbalance and the migration cost. 
\citet{rivetti2015efficient} propose an algorithm for a similar problem with an offline training phase to learn the best explicit mapping.
Comparatively, we propose a light-weight streaming algorithm that provide strong load balancing guarantees, and which does not require training or active monitoring of the load imbalance.
We show that by splitting the head of the distribution on multiple workers, and the tail distribution on at-most two workers, the techniques in this paper achieve nearly ideal load balance.

The load balancing problem has been extensively studied in the theoretical research community~\cite{azar1999balanced-allocations,talwar2007weightedcase,wieder2007heterogeneousbins}.
There have also been several proposals to solve the load balancing problem for peer-to-peer systems~\citep{rao2003load,karger2004simple}.
The E-Store system explores the idea of handling hot tuples separately from cold tuples in the context of an elastic database management system~\citep{taft2014store}.
Most existing load balancing techniques for \dspes are analogous to key grouping with rebalancing~\citep{shah2003flux,cherniack2003scalable,xing2005dynamic,balkesen2013adaptive,castro2013integrating}.
Flux monitors the load of each operator, ranks servers by load, and migrates operators from the most loaded to the least loaded servers~\citep{shah2003flux}.
Aurora* and Medusa propose policies to operator migration in \dspes and federated \dspes~\citep{cherniack2003scalable}.
Borealis uses a similar approach, but also aims at reducing the correlation of load spikes among operators placed on the same server~\citep{xing2005dynamic}.
This correlation is estimated by using load samples taken in the recent past.
Similarly, \citet{ balkesen2013adaptive} propose a frequency-aware hash-based partitioning to achieve load balance.
\citet{castro2013integrating} propose integrating common operator state management techniques for both checkpointing and migration.
Differently from these techniques, the ones proposed in this work do not require any rebalancing, operator migration, or explicit mapping of the keys to servers (which necessitates potentially huge routing tables).

\section{Conclusion}
We studied the problem of load balancing for distributed stream processing engines when deployed at large scale.
In doing so, we analyzed how highly skewed data streams exacerbate the problem, and how this affects the throughput and latency of an application.
We showed that existing state-of-the-art techniques fall short in these extreme conditions.

We proposed two novel techniques for this tough problem: D-Choices and W-Choices.
These techniques employ a streaming algorithm to detect heavy hitter for tracking the hot keys in the stream, which constitute the head of the distribution of keys, and allows those hot keys to be processed on larger set of workers.
In particular, W-Choices allows a head key to be processed on entire set of workers, while D-Choices places a key on a smaller number of workers.
This number is determined by using an easy-to-compute lower bound derived analytically and verified empirically.

We evaluated our proposal via extensive experimentation that covers simulations with synthetic and real datasets, and deployment on a cluster running Apache Storm.
Results show that the techniques in this paper achieve better load balance compared to the state-of-the-art approaches.
This improvement in balance translates into a gain in throughput of up to $150\%$, and latency of up to $60\%$ over \pkgs for the cluster setup.
When compared to key grouping, the gains are even higher: $230\%$ for throughput and $75\%$ for latency.
Overall, D-Choices and W-Choices achieve the same throughput and latency as shuffle grouping, at just a fraction of the cost in terms of memory.


\bibliographystyle{IEEEtran}
\balance
\bibliography{references}
\begin{appendices}
\section{Expected Size of a Worker Set}
\label{sec:collisions}
The process of chosing the workers for a key, is akin to placing $d$ items, uniformly at random and with replacement, into \numworkers possible slots, with $d < \numworkers$ (in our case, if $d = \numworkers$ we should switch to the W-Choices algorithm).
We now derive the expected number of slots $b$ filled with at least one item.

As usual, it is easier to answer the complementary question: how many slots are left empty?
Each slot has probability $p_{\text{empty}} = \left( 1 - \frac{1}{\numworkers} \right) ^ d$ to be left empty at the end of $d$ independent placements.
Therefore, the expected number of empty slots is
\[
\expect[X_{\text{empty}}] = \sum_{i = 1}^{\numworkers} p_{\text{empty}} = \numworkers \left( \frac{\numworkers - 1}{\numworkers} \right) ^ d .
\]
The number of full slots can be computed by subtracting the number of empty slots from the total number of slots,
\[
b = \expect[X_{\text{full}}] = \numworkers - \expect[X_{\text{empty}}] = \numworkers - \numworkers \left( \frac{\numworkers - 1}{\numworkers} \right) ^ d .
\]
The expression of $b_h$ can be obtained by replacing $d$ with $h \times d$ in the previous expression, since we want to place $h \times d$ items.


\end{appendices}
\end{document}